\documentclass[11pt]{article}
\usepackage[in]{fullpage}
\usepackage[utf8]{inputenc} 
\usepackage{url}            
\usepackage{booktabs}       
\usepackage{amsfonts}       
\usepackage{nicefrac}       
\usepackage{microtype}      
\usepackage{amsmath, amsthm, amssymb}
\usepackage{fancybox}
\usepackage{tkz-graph}
\usetikzlibrary{arrows}
\usepackage[noend]{algpseudocode}
\usepackage{algorithm,algorithmicx}
\usepackage[symbol]{footmisc}
\usepackage{tikz}
\usepackage{enumitem}
\usepackage{xspace}%
\usepackage{xcolor}
\usepackage{multicol}

\usepackage{subcaption}

\usepackage{hyperref}
\hypersetup{%
   breaklinks,%
   colorlinks=true,%
   linkcolor=[rgb]{0.45,0.0,0.0},%
   citecolor=[rgb]{0,0,0.45}
}

\tikzset{
    every node/.style={
        circle,
        draw,
        fill          = black!50,
        inner sep     = 0pt,
        minimum width =4 pt
    }   
}  

\numberwithin{figure}{section}%
\numberwithin{table}{section}%
\numberwithin{equation}{section}%

  \theoremstyle{plain}%
 \newtheorem{theorem}{Theorem}[section]
 
 \newtheorem{lemma}[theorem]{Lemma}%

 \newtheorem{corollary}[theorem]{Corollary}
  \newtheorem{problem}{Problem}[section]%

  \newtheorem{proposition}[theorem]{Proposition}
  
  \newtheorem{remark}[theorem]{Remark}
  
\newcommand{\mvd}{MR$(K_n, \Re)$\xspace}

\newcommand{\mvid}{MR$(K_n, \Re_{\geq 0})$\xspace}
\newcommand{\gmd}{MR$(G, \Re)$\xspace}
\newcommand{\gmvd}{MR$(G, \Re)$\xspace}

\newcommand{\gmvid}{MR$(G, \Re_{\geq 0})$\xspace}
\newcommand{\mcut}{\textsf{MULTICUT}\xspace}
\newcommand{\lcut}{\textsf{LB-CUT}\xspace}

\newcommand{\csp}{\#\mathsf{sp}}
\newcommand{\deficit}{\kappa}

\newcommand{\opt}{OPT}

\renewcommand{\Re}{\mathbb{R}}%
\DefineNamedColor{named}{RedViolet} {cmyk}{0.07,0.90,0,0.34}

\newcommand{\eps}{{\varepsilon}}%
\newcommand{\atgen}{\symbol{'100}}

\newcommand{\UMichThanks}[1]{\thanks{Department of Mathematics, University of Michigan - Ann Arbor, {\tt \{annacg, rsonthal\}\atgen{}umich.edu}.#1}}
      
\newcommand{\UTDThanks}[1]{\thanks{Department of Computer Science;
      University of Texas at Dallas; 
      {\tt \{cxf160130, benjamin.raichel, greg.vanbuskirk\}\atgen{}utdallas.edu}. #1}}
      
\newcommand{\myparagraph}[1]{\bigskip\noindent{\textbf{#1}}}

\newenvironment{reftheorem}[1]{\medskip\parindent 0pt{\bf Theorem \ref{#1}.}\em }{\vspace{1em}}

\newcommand{\Sym}{\text{Sym}_n(\mathbb{R}_{\ge 0})}

\DeclareMathOperator*{\argmin}{arg\,min}

\DeclareCaptionSubType*{algorithm}

\title{Generalized Metric Repair on Graphs\footnote{This paper combines and significantly extends the results from \cite{frv-mvdre-18} and \cite{gs-gmrg-18}.}}
\date{}

\author{
 Chenglin Fan\UTDThanks{Work on this paper was partially
      supported by NSF CRII Award 1566137 and CAREER Award 1750780.}
 \and 
 Anna C. Gilbert\UMichThanks{}
 \and
 Benjamin Raichel\footnotemark[2]
 \and 
 Rishi Sonthalia\footnotemark[3]
     \and
 Gregory Van Buskirk\footnotemark[2]
}

\newcommand{\remove}[1]{}

\begin{document}

\maketitle
\thispagestyle{empty}

\begin{abstract}
Many modern data analysis algorithms either assume or are considerably more efficient if the distances between the data points satisfy a metric. These algorithms include metric learning, clustering, and dimension reduction. As real data sets are noisy, distances often fail to satisfy a metric. For this reason, Gilbert and Jain~\cite{Gilbert2017} and Fan et al.~\cite{frv-mvdha-18} introduced the closely related \emph{sparse metric repair} and \emph{metric violation distance} problems. The goal of these problems is to repair as few distances as possible to ensure they satisfy a metric. Three variants were considered, one admitting a polynomial time algorithm. The other variants were shown to be APX-hard, and an $O(OPT^{1/3})$-approximation was given, where $OPT$ is the optimal solution size.

In this paper, we generalize these problems to no longer consider \emph{all} distances between the data points. That is, we consider a weighted graph $G$ with corrupted weights $w$, and our goal is to find the smallest number of weight modifications so that the resulting weighted graph distances satisfy a metric. This is a natural generalization and is more flexible as it takes into account different relationships among the data points. As in previous work, we distinguish among the types of repairs permitted and focus on the increase only and general versions. We demonstrate the inherent combinatorial structure of the problem, and give an approximation-preserving reduction from \mcut, which is hard to approximate within any constant factor assuming the Unique Games Conjecture. 
Conversely, we show that for any fixed constant $\varsigma$, for the large class of $\varsigma$-chordal graphs, the problems are fixed parameter tractable, answering an open question from previous work. 
Call a cycle \emph{broken} if it contains an edge whose weight is larger than the sum of all its other edges, and call the amount of this difference its \emph{deficit}. We present approximation algorithms, one which depends on the maximum number of edges in a broken cycle, and one which depends on the number of distinct deficit values, both quantities which may naturally be small. 
Finally, we give improved analysis of previous algorithms for complete graphs.
\end{abstract}

\newpage
\pagenumbering{arabic}



\section{Introduction}

Similarity measures that satisfy a metric are fundamental to a large number of machine learning tasks such as dimensionality reduction and clustering (see~\cite{Wang2015, Baraty2011} for two examples). However, due to noise, missing data, and other corruptions, in practice these distances often do not adhere to a metric. There are also many algorithmic settings where the underlying distances arise from a metric space or are at least well modeled by one. Such cases are fortuitous, as certain tasks become provably easier over metric data (e.g., approximating the optimal TSP tour), and moreover they allow us to use a number of computational tools such as metric embeddings. Motivated by these observations and the earlier work of Brickell et~al.\ \cite{Brickell}, Fan et~al.\ \cite{frv-mvdha-18} and Gilbert and Jain \cite{Gilbert2017} respectively formulated the \emph{Metric Violation Distance (MVD)} and the \emph{Sparse Metric Repair (SMR)} problems. Formally, the problem both sets of authors studied was: given a full distance matrix, modify as few entries as possible so that the repaired distances satisfy a metric.

To capture a more general nature of the problem, we define the \emph{Graph Metric Repair} problem as the natural graph theoretic generalization of the MVD and SMR problems: 
\begin{quote} Given a positively weighted undirected graph $G=(V,E,w)$ and a set $\Omega\subseteq \mathbb{R}$, find the smallest set of edges $S\subseteq E$ such that by modifying the weight of each edge in $S$, by adding a value from $\Omega$, the new distances satisfy a metric.\end{quote}

This additional graph structure introduced in the generalized problem lets us incorporate different types of relationships amongst data points and gives us more flexibility in its structure, and hence avails itself to be applicable to a richer class of problems. Gilbert and Sonthalia~\cite{GilbertSonthalia:SwissRoll2018} use (graph) metric repair to learn metrics and metric embeddings in data sets with missing data. Furthermore, while Gilbert and Jain \cite{Gilbert2017} showed that SMR can be approximated empirically via convex optimization, both \cite{Gilbert2017} and \cite{frv-mvdha-18} developed combinatorial algorithms based upon All Pairs Shortest Path (APSP) computations. Thus, metric repair is inherently a combinatorial problem and the generalized graph problem helps elucidate this structure. 

Graph Metric Repair is related to a large number of other previously studied problems.  
A short list includes:
metric nearness, seeking the metric minimizing the sum of distance value changes \cite{Brickell}; 
metric embedding with outliers, seeking the fewest points whose removal creates a metric \cite{sww-meo-17}; 
matrix completion, seeking to fill missing matrix entries to produce a low rank \cite{cr-emcco-12}; and many more.  
See \cite{frv-mvdha-18} for a more detailed discussion of these and other problems.  

Here we consider the deep connections to cutting problems, which underlie several results in this paper, 
and which were not previously observed in \cite{frv-mvdha-18,Gilbert2017}. 
In particular, our problem is closely related to \mcut, a generalization of the standard $s$-$t$ cut problem to multiple $s$-$t$ pairs.
\mcut has been extensively studied, both for directed and undirected graphs.
For undirected graphs, the problem captures vertex cover even when $G$ is a tree. 
Moreover, assuming the Unique Games Conjecture (UGC) there is no constant factor approximation \cite{ChawlaKKRS06}.
In general, the best known approximation factor is $O(\log k)$~\cite{GargVY96}, for $k$ terminal pairs,
which improves to an $O(r)$-approximation when $G$ excludes $K_r$ as a minor \cite{AbrahamGGNT14}.
%
%
Another closely related problem is Length Bounded Cut (\lcut), where given a value $L$ and an $s$-$t$ pair, 
the goal is to delete the minimum number of edges such that there is no path between $s$ and $t$ with length $\leq L$. 
\lcut is hard to approximate within a factor of $\Omega(\sqrt{L})$ in undirected graphs \cite{l-ihcifp-17}.

\myparagraph{Contributions and Results:}
The main contributions of this paper are as follows: 
\begin{itemize}[nosep]
\item We show the decrease only version of the problem ($\Omega = \mathbb{R}_{\le 0}$) is solvable in cubic time,  
and that if distances are allowed to increase even by a single number, the problem is NP-Complete. 
\item We provide a \emph{characterization} for the support of solutions to the increase ($\Omega = \mathbb{R}_{\ge 0}$) and general ($\Omega = \mathbb{R}$) versions of the problem. This characterization is fundamental and is the basis for the rest of our results. Furthermore, we provide a cubic time algorithm determining for any given subset of the edges whether there exists \emph{any valid solution} with that support, and finds one if it exists. 
Additionally, we show the increase only problem reduces to the general one.
\item We give polynomial-time approximation-preserving reductions from \mcut and \lcut to graph metric repair. 
This connection to the well studied \mcut problem is interesting in its own right, but also implies 
graph metric repair is NP-hard, and cannot be approximated within any constant factor assuming UGC. 
Our reduction from \lcut implies that, for any fixed $L$, the set of instances of graph metric repair with maximum edge weight $L$  (and minimum weight 1) are hard to approximate within a factor of $\Omega(\sqrt{L})$.
\item For any fixed constant $\varsigma$, by parameterizing on the size of the optimal solution, we present a \emph{fixed parameter tractable} algorithm for the case when $G$ is $\varsigma$-chordal. This not only answers an open question posed by~\cite{frv-mvdha-18} for complete graphs, 
but significantly extends it to the larger $\varsigma$-chordal case (see~\cite{Chandran2005GraphsOL} for characterizations of such graphs, many of which are the complements of a variety of families of graphs). Our FPT algorithm requires a number of new and interesting insights into the structure of the metric repair problem. Moreover, we get an upper bound on the number of optimal supports, as each one is seen by some branch of the algorithm.
\item We give several approximation algorithms, parameterized by different measures of how far the input is from a metric.
Call a cycle \emph{broken} if it contains an edge whose weight is larger than the sum of all its other edges, and call the amount of this difference its \emph{deficit}.
First, we argue that our characterization of optimal supports directly implies an $L$-approximation, where $L+1$ is the largest number of edges in a broken cycle.
Next, by analyzing the structure of the problem more carefully, we give an $O(\deficit \log n)$-approximation, where $\kappa$ is the number of distinct positive cycle deficit values.
While in general $\deficit$ may be large, when it is small it still allows for graphs with large chordless broken cycles, a case not handled by our other algorithms.
Significantly, our approximations mirror our hardness results. 
We give an $L$-approximation, while $\lcut$ gives $\Omega(\sqrt{L})$-hardness.
We give an $O(\deficit \log n)$-approximation, while in general the best known approximation for \mcut is $O(\log n)$.
\item Finally, we give improved analysis of previous algorithms for the complete graph case.
To keep the focus on our main results, this entire section has been moved to Appendix \ref{apnd:completeGraphs}.
\end{itemize}


\section{Preliminaries}\label{sec:prelim}

\subsection{Notation and problem definition}

Let us start by defining some terminology.
Throughout the paper, the input is an undirected and weighted graph $G=(V,E,w)$. 
A subgraph $C=(V',E')$ is called a $k$-cycle if $|V'|=|E'|=k$, and the subgraph is connected with every vertex having degree exactly $2$. 
We often overload this notation and use $C$ to denote either the cyclically ordered list of vertices or edges from this subgraph. 
Let $C\setminus e$ denote the set of edges of $C$ after removing the edge $e$, and $\pi(C\setminus e)$ denote the corresponding induced path between the endpoints of $e$.

A cycle $C$ is \textbf{broken} if there exists an edge $h\in C$ such that 
\[ w(h) > \sum_{e \in C \backslash h} w(e) \] 
In this case, we call the edge $h$ the \textbf{heavy} edge of $C$, and all other edges of $C$ are called \textbf{light} edges. 
We call a set of edges a \textbf{light cover} if it contains at least one light edge from each broken cycle. 
Similarly, we call it a \textbf{regular cover} if it contains at least one edge from each broken cycle.
We say that a weighted graph $G = (V,E,w)$ \textbf{satisfies a metric} if there are no broken cycles. 
Finally, let {\rm Sym}$_n(\Omega)$ be the set of $n \times n$ symmetric matrices with entries drawn from $\Omega\subseteq \mathbb{R}$. 
Note that the weight function $w$ can be viewed as an $n \times n$ symmetric matrix (missing edges get weight $\infty$), and thus for any $W\in {\rm Sym}_n (\Omega)$, the matrix sum $w+W$ defines a new weight function.
Now we can define the generalized graph metric repair problem as follows. In the following, $\|W\|_0$ is the number of non-zero entries in the matrix $W$, 
i.e., the $\ell_0$ pseudonorm when viewing the matrix $W$ as a vector. 

\begin{problem} \label{prob:metric} 
Given a set $\Omega\subseteq \mathbb{R}$ and a positively weighted graph $G = (V,E,w)$ we want to find
\[ 
	\underset{W \in {\rm Sym}(\Omega)}{\argmin}\|W\|_0 \text{ such that } G = (V,E,w+W) \text{ satisfies a metric, or return NONE,} 
\] 
if no such $W$ exists. 
Denote this problem as graph metric repair or MR($G, \Omega$). 
\end{problem}

A matrix $W$ is an \emph{optimal solution} if it realizes the $\arg \min$ in the above, and is a \emph{solution} (without the \emph{optimal} prefix) if $G = (V,E,w+W)$ satisfies a metric, but $\|W\|_0$ is not required to be minimum.
The \textbf{support} of a matrix $W \in {\rm Sym}(\Omega)$, denoted $S_W$, is the set of edges corresponding to non-zero entries in $W$.
As we will see in Proposition \ref{prop:verifier}, given a support for a solution $W$, we can easily find satisfying entries.
Thus, the main difficulty lies in finding the support.
Throughout we use $\opt$ to denote the size of the support of an optimal solution.

We also need the following basic graph theory definitions: $K_n$ is the complete graph on $n$ vertices. $C_n$ is the cycle $n$ vertices. A \textbf{chord} of a cycle is an edge connecting two non-adjacent vertices. 
For a given value $\varsigma$, a graph $G$ is called a $\varsigma$-\textbf{chordal} if the size of the largest chordless cycle in $G$ is $\leq \varsigma$. 

Let the \textbf{deficit} of a broken cycle $C$, denoted $\delta(C)$, be the weight of its heavy edge minus the sum of the weights of all other edges in $C$.  
Similarly, $\delta(G)$ denotes the maximum of $\delta(C)$ over all broken cycles.
Finally, let $L+1$ be the maximum number of edges in a broken cycle (i.e., $L$ counts the light edges). 
Note $\delta$ and $L$ are both parameters measuring the extent to which cycles are broken, $\delta$ with respect to the weights and $L$ with respect to the number of edges.

\subsection{Previous results}

Fan et al.~\cite{frv-mvdha-18} and Gilbert and Jain~\cite{Gilbert2017} studied the special case of MR$(G, \Omega)$ where $G = K_n$.
Three sub-cases based on $\Omega$ were considered, namely $\Omega = \mathbb{R}_{\le 0}$ (decrease only), $\mathbb{R}_{\ge 0}$ (increase only), and $\mathbb{R}$ (general). 
Various structural, hardness, and algorithmic results were presented for these cases.
In particular, the major results from these previous works are as follows. (Note the notation and terminology here differs slightly from \cite{frv-mvdha-18, Gilbert2017}.)

\begin{theorem} \cite{frv-mvdha-18, Gilbert2017} \label{thm:specialDecrease} The problem MR$(K_n, \mathbb{R}_{\le 0})$ can be solved in cubic time. 
\end{theorem}

\begin{theorem} \cite{frv-mvdha-18} \label{thm:specialStrctr} For a complete positively weighted graph $K_n= (V,E,w)$ and $S \subseteq E$ we have: 
\begin{enumerate}[nosep]
\item $S$ is a regular cover if and only if $S$ is the support to a solution to MR$(K_n, \mathbb{R})$. 
\item $S$ is a light cover if and only if $S$ is the support to a solution to MR$(K_n, \mathbb{R}_{\ge 0})$.
\end{enumerate} 
\end{theorem}

%

\begin{theorem} \cite{frv-mvdha-18,Gilbert2017} \label{thm:oracle}
Given the support $S$ of a solution to either MR$(K_n, \mathbb{R}_{\ge 0})$ or MR$(K_n, \mathbb{R})$, in polynomial time one can find a weight assignment to the edges in $S$ which is a solution.

\cite{Gilbert2017} Moreover, for MR$(K_n, \mathbb{R}_{\ge 0})$, if $K_n\!-\!S$ is connected, then for any edge $uv \in S$, setting the weight of $uv$ to be the shortest distance between $u$ and $v$ in $K_n\!-\!S$ is a solution.
\end{theorem}

\begin{theorem}\cite{frv-mvdha-18} \label{thm:specialNP} 
The problems MR$(K_n, \mathbb{R}_{\ge 0})$ and MR$(K_n, \mathbb{R})$ are APX-Complete, and moreover permit $O(OPT^{1/3})$ approximation algorithms. 
\end{theorem}

\section{Transitioning to Graph Metric Repair} 
In this section we generalize theorems \ref{thm:specialDecrease}, \ref{thm:specialStrctr}, and \ref{thm:oracle} to the case when $G$ is any graph, and additionally show that for general graphs \gmvid reduces to \gmvd.
Subsequently, in the later sections of paper, we provide a number of new stronger hardness and approximation results for MR$(G, \mathbb{R}_{\ge 0})$ and MR$(G, \mathbb{R})$ for general graphs, 
as well as an FPT algorithm for $\varsigma$-chordal graphs, in effect generalizing and strengthening Theorem \ref{thm:specialNP}, and answering previously unresolved questions.

For MR$(G,\mathbb{R}_{\le 0})$ we have the following generalization of Theorem \ref{thm:specialDecrease}. 
Moreover, we observe the hardness proof of \cite{frv-mvdha-18} implies if weights are allowed to increase even by a single value, the problem is APX-Complete.  
The proof of the theorem below follows fairly directly from previous work, and so has been moved to Appendix \ref{apnd:decrease}, which contains additional corollaries.

\begin{theorem}\label{thm:decreaseNew}
The problem MR$(G, \mathbb{R}_{\le 0})$ can be solved in $O(n^3)$ time. 

Moreover, the problem becomes hard if even a single positive value is allowed.  That is, if $0 \in \Omega$ and $\Omega \cap \mathbb{R}_{> 0} \neq \emptyset$ then MR$(G, \Omega)$  is APX-Complete.

\end{theorem}

\subsection{Structural results}

Theorem \ref{thm:specialStrctr} suggests that the problem is mostly combinatorial in nature. We shall see that, in general, the difficult part of the problem is finding the support of an optimal solution. 
Next, we present a characterization of the support of all solutions to the graph metric repair problem, generalizing Theorems \ref{thm:specialStrctr}, \ref{thm:oracle}. It should be noted the following proof is significantly simpler than the proof of Theorem \ref{thm:specialStrctr} in \cite{frv-mvdha-18}. The key insight in the generalization is: 
\begin{enumerate}[label=(\roman*),series=insights] 
\item If the shortest path between two adjacent vertices is the not the edge connecting them, then this edge is the heavy edge of a broken cycle. \label{item:broken}
\end{enumerate}

\begin{theorem} \label{thm:structure} 
For any positively weighted graph $G = (V,E,w)$ and $S \subseteq E$, the following hold:
\begin{enumerate}[nosep]
\item $S$ is a regular cover if and only if $S$ is the support to a solution to MR$(G, \mathbb{R})$. 
\item $S$ is a light cover if and only if $S$ is the support to a solution to MR$(G, \mathbb{R}_{\ge 0})$.
\end{enumerate} 
\end{theorem}

\begin{proof}
First, assume that $S$ is the support of a solution to MR$(G, \mathbb{R})$ (MR$(G, \mathbb{R}_{\ge 0})$). Suppose $C$ is a broken cycle in $G$. If $S$ does not contain any (light) edges from $C$, then changing (increasing) the weights on $S$ could not have fixed $C$. Hence, $S$ must be a regular (light) cover thus proving the ``if'' direction of both parts of the theorem.

For the ``only if'' direction, we are given a regular (light) cover $S\subseteq E$ which we use to define a graph $\hat{G} = (V,E\setminus S,w)$.
Note that since $S$ is either a regular or light cover, $S$ contains at least one edge from all broken cycles of $G$.
Thus, since $\hat{G}$ is $G$ with the edges of $S$ removed, $\hat{G}$ has no broken cycles.
Therefore, the shortest path between all adjacent vertices in $\hat{G}$ is the edge connecting them.

Now we define another graph $G'=(V,E,w')$ where $w'(e) = w(e)$ for all $e \in E \setminus S$ and for all $e \in S$, $w'(e)$ is the length of the shortest path between its end points in $\hat{G}$ or $\|w\|_\infty$ (the maximum edge weight in $\hat{G}$) if no path exists. 

To prove \textbf{1.}, it suffices to show $G'$ satisfies a metric, since $G'$ is $G$ with only weights from edges in $S$ modified.
For any edge $e \in E$, if $w'(e)$ is the shortest path between its nodes in $G'$ then $e$ is not a heavy edge in $G'$.
Therefore, edges that are in both $G'$ and $\hat{G}$ and edges that are in $G'$ whose weight was set to length of the shortest path between its end points in $\hat{G}$ are not heavy edges.  Thus, we only need to look at edges in $G'$ whose weight is $\|w\|_\infty$. These are edges that connect two disconnected components in $\hat{G}$. Thus, any cycle in $G'$ with such an edge must involve another edge between components which also has weight $\|w\|_\infty$. However, a cycle with two edges of maximum weight cannot be broken, and thus such edges cannot be heavy edges in $G'$.
Therefore, there are no heavy edges in $G'$, and so $G'$ satisfies a metric.

To prove \textbf{2.}, it now suffices to show that for all $e \in E$, we have that $w'(e) \ge w(e)$. For all $e \in E\setminus S$, we know that $w'(e) = w(e)$. Now, suppose for contradiction that for some $e \in S$, we have $w'(e) < w(e)$. Note if we set $w'(e) = \|w\|_\infty$, then we cannot have $w'(e) < w(e)$. Thus, $w'(e)$ must be the weight of the shortest path between the end points of $e$ in $\hat{G}$. Let $P$ be this shortest path in $\hat{G}$. This implies $G$ has a broken cycle $C = P \cup \{e\}$ for which $e$ is the heavy edge. Since $S$ is a light cover, it has a light edge from each broken cycle. So, $S$ must have a light edge from $C$, but then $P$ could not have existed in $\hat{G}$, a contradiction. Hence, $w'(e) \ge w(e)$ and we have an increase only solution with such a set $S$. 
\end{proof}

Furthermore, given a weighted graph $G$ and a potential support $S_W$ for a solution $W$, in $O(n^3)$ time we can determine whether there exists a valid (increase only or general) solution on that support, and if so, find one. This is a generalization of Theorem~\ref{thm:oracle}, improving upon the linear programming approach of \cite{frv-mvdha-18}. Its proof is related to the above theorem, and again uses insight \ref{item:broken}.

\begin{algorithm}[ht!]
\caption{Verifier}
\label{alg:ver}
\begin{algorithmic}[1]
\Function{Verifier}{$G=(V,E,w),S$}
\State $M = \|w\|_\infty$, $\hat{G} = (V,E, \hat{w})$
\State For each $e \in S$ set $\hat{w}(e) = M$ and for each $e \in E\setminus S$, set $\hat{w}(e) = w(e)$
\State For each $(u,v) \in E$, update $w(u,v)$ to be length of the shortest path from $u$ to $v$ in $\hat{G}$\label{line:update}
\If{Only edges in $S$ had weights changed (or increased for increase only case) }
\State \Return $w$
\Else
\State \Return NULL 
\EndIf
\EndFunction
\end{algorithmic}
\end{algorithm}

\begin{proposition}
\label{prop:verifier} 
The {\sc Verifier} algorithm, given a weighted graph $G$ and a potential support for a solution $S$, determines in $O(n^3)$ time whether there exists a valid (increase only or general) solution on that support and if so finds one. 
\end{proposition}

\begin{proof}
Let $G=(V,E,w)$ be the original graph and let $M$ be the maximum edge weight from the graph $G$. The algorithm defines a new graph $\hat{G}=(V,E,\hat{w})$, with the following weights \[ \hat{w}(e)= \begin{cases} w(e) & e \not\in S \\ M & e \in S \end{cases} \] For each $e = (v_1,v_2) \in E$, line \ref{line:update} sets $w(e)$ to be the weight of the shortest path in $\hat{G}$ from $v_1$ to $v_2$. 
Thus, at the end of the algorithm $w(e)$ satisfies the shortest path metric of $\hat{G}$. As the algorithm outputs $w$ if and only if only edge weights in $S$ are modified (increased), it suffices to argue $S$ is a regular cover (light cover) if and only if only edge weights in $S$ are modified (increased).

Assume that $S$ is a regular or light cover. We argue line \ref{line:update} only updates the weights of the edges in $S$. Note that $G \setminus S$ has no broken cycles. Thus, for any $e = (v_1,v_2) \in G\setminus S$ we have that the shortest path from $v_1$ to $v_2$ must be $e$. Now consider any path $P$ from $v_1$ to $v_2$ in $\hat{G}$. If $P\cap S = \emptyset$, then $w(P) \ge w(e)$. On the other hand if $P \cap S \neq \emptyset$, then let $\tilde{e} \in P\cap S$. Then, we have that \[ w(P) \ge w(\tilde{e}) = M \ge w(e) \] 
Thus, in either case, $w(P) \ge w(e)$. Hence for all $e \in G\setminus S$ we do not change its weight. 

If $S$ is a light cover, we also need to argue that the weights only increased. Let $e =(v_1,v_2) \in S$. Let $P$ be a path of smallest weight in $\hat{G}$. Suppose $P \cap S \neq \emptyset$, then, we have that  $w(P) \ge M \ge w(e)$. Thus, in this case we could not have decreased the weight. Thus, assume that  $P \cap S = \emptyset$. If we still have that $w(P) \ge w(e)$, then we could not have decreased the weight. Thus, let us further assume that $w(P) < w(e)$. In this case, $P$ along with $e$ form a broken cycle in $G$, with $e$ as the heavy edge. But then since $S$ is a light cover, we have that $P \cap S \neq \emptyset$. Thus, we have a contradiction and this case cannot occur. Thus, if $S$ is a light cover, then we only increase the edge weights. 

Now assume $S$ is not a regular cover (light cover). Then there exists a broken cycle $C$ such that none of its (light) edges are in $S$. Thus, there is a broken cycle $C$ in $\hat{G}$. Let $e$ be the heavy edge of $C$, then on line \ref{line:update} the weight of $e$ will be decreased, and thus our algorithm will return NULL.
\end{proof}

The next theorem shows that once we know the support, the set of all possible solutions on that support is a nice space. 
 
\begin{theorem}
For any weighted graph $G$ and support $S$ we have that the set of solutions with support $S$ is a closed convex subset of $\mathbb{R}^{n \times n}$. Additionally, if $G - S$ is a connected graph or we require an upper bound on the weight of each edge, then the set of solutions is compact.
\end{theorem}
\begin{proof} Let $x_{ij}$ for $1 \le i,j \le n$ be our coordinates. Then the equations $x_{ij} = c_{ij}$ for $(i,j)$ not in the support and $x_{ij} \le x_{ik}+x_{kj}$ define a closed convex set. Thus, we see the first part. For the second part we just need to see that set is bounded to get compactness. If we have that $G - S$ is connected then for all $e \in S$ there is a path between end points of $e$ in $G-S$. Thus, the weight of this path is an upper bound. On the other hand 0 is always a lower bound. Thus, we get compactness if $G-S$ is connected.\end{proof}

\subsection{Reducing MR$(G, \mathbb{R}_{\ge 0})$ to MR$(G, \mathbb{R})$}

We now show that \gmvid reduces to \gmvd. In later sections, this lets us focus on \gmvd for our algorithms and \gmvid for our hardness results.  Note that whether an analogous statement holds for the previously studied complete graph case, $G=K_n$, is not known, and the following does not immediately imply this as it does not construct a complete graph.

\begin{theorem}\label{thm:gmvid_to_gmvd}
There is an approximation-preserving, polynomial-time reduction from \gmvid to \gmvd.
\end{theorem}
\begin{proof}
Let $G=(V,E,w)$ be an instance of \gmvid. 
Find the set $H=\{(s_1,t_1),\ldots, \allowbreak (s_{|H|},t_{|H|})\}$ of heavy edges of all broken cycles by comparing the weight of each edge to the shortest path distance between its endpoints. 
We now construct an instance, $G'=(V',E',w)$, of \gmvd. 
For all $1\leq i\leq |H|$ and $1\leq j\leq |E|+1$, let $Q = \{v_{ij}\}_{i,j}$ be a vertex set, and let $F_l = \{ (s_i, v_{ij}) \}_{i,j}$ and $F_r = \{ (t_i, v_{ij}) \}_{i,j}$ be edge sets.
Let $V' = V \cup Q$ and $E'=E \cup F_l \cup F_r$, where all $(s_i, v_{ij})$ edges in $F_l$ have weight $Z=1 + \max_{e\in E} w(e)$, 
and for any $i$ all $(t_i, v_{ij})$ edges in $F_r$ have weight $Z - w((s_i, t_i))$.

Let $C$ be any broken cycle in $G$ with heavy edge $(s_i,t_i)$ for some $i$. 
First, observe that the cycle $C'=(C\setminus (s_i,t_i))\cup\{(s_i, v_{ij}), (t_i, v_{ij})\}$ is a broken cycle with heavy edge $(s_i, v_{ij})$, for any $j$.
To see this, note that $w((s_i, v_{ij}))=Z=w((t_i, v_{ij}))+w((s_i, t_i))$. Thus since $C$ is broken, 
\[
w((s_i, v_{ij}))=w((t_i, v_{ij}))+w((s_i, t_i)) > w((t_i, v_{ij}))+ w(C\setminus (s_i,t_i)),
\]
and thus by definition $C'$ is broken with heavy edge $(s_i, v_{ij})$.
Hence each broken cycle $C$ in $G$, with heavy edge $(s_i,t_i)$, corresponds to $|E|+2$ broken cycles in $G'$, 
namely, $C$ itself and the cycles obtained by replacing $(s_i, t_{i})$ with a pair $(s_i, v_{ij}),(t_i, v_{ij})$, for any $j$.

We now show the converse, that any broken cycle $C'$ in $G'$ is either also a broken cycle $C$ in $G$, 
or obtained from a broken cycle $C$ in $G$ by replacing $(s_i, t_{i})$ with $(s_i, v_{ij}),(t_i, v_{ij})$ for some $j$. 
First, observe that for any $i$, any cycle containing the edge $(s_i,v_{ij})$ must also contain the edge $(t_i,v_{ij})$, 
and moreover, if a cycle containing such a pair is broken, then its heavy edge must be $(s_i,v_{ij})$ as $w((s_i,v_{ij}))=Z$.
Similarly, any cycle containing more than one of these pairs of edges (over all $i$ and $j$) is not broken, since such cycles then would contain at least two edges with the maximum edge weight $Z$.
So let $C'$ be any broken cycle containing exactly one such $(s_i,v_{ij})$, $(t_i,v_{ij})$ pair. 
Note that $C'$ cannot be the cycle $((s_i,v_{ij}),(t_i,v_{ij}),(s_i,t_i))$, as this cycle is not broken because $w((s_i, v_{ij}))=w((t_i, v_{ij}))+w((s_i, t_i))$.
Thus, $C=C'\setminus\{(s_i,v_{ij}),(t_i,v_{ij})\} \cup \{(s_i,t_i)\}$ is a cycle, and $C'$ being broken implies $C$ is broken with heavy edge $(s_i,t_i)$, 
implying the claim. This holds since 
\[
w(s_i,t_i) = w((s_i, v_{ij})) - w((t_i, v_{ij})) > w(C'\setminus (s_i, v_{ij})) - w((t_i, v_{ij})) = 
w(C\setminus (s_i,t_i)).
\]

Now consider any optimal solution $M$ to the \gmvid instance $G$, which by Theorem \ref{thm:structure} we know is a minimum cardinality light cover of $G$. 
By the above, we know that $M$ is also a light cover of $G'$, and hence is also a regular cover of $G'$.
Thus by Theorem \ref{thm:structure}, $M$ is a valid solution to the \gmvd instance.  
Conversely, consider any optimal solution $M'$ to the \gmvd instance $G'$, which by Theorem \ref{thm:structure} is a minimum cardinality regular cover of $G'$.
The claim is that $M'$ is also a light cover of $G$, and hence is a valid solution to the \gmvid instance. 
To see this, observe that since all broken cycles in $G$ are broken cycles in $G'$, $M'$ must be a regular cover of all broken cycles in $G$, 
and we now argue that it is in fact a light cover. 
Specifically, consider all the broken cycles in $G$ which have a common heavy edge $(s_i,t_i)$.  
Suppose there is some cycle in this set, call it $C$, which is not light covered by $M'$. 
As $M'$ is a regular cover for $G'$, this implies that for any $j$, the broken cycle described above
determined by removing the edge $(s_i, t_i)$ from $C$ and adding edges $(s_i,v_{ij})$ and $(t_i,v_{ij})$, 
must be covered either with $(s_i,v_{ij})$ or $(t_i,v_{ij})$.  
However, as $j$ ranges over $|E|+1$ values, and these edge pairs have distinct edges for different values of $j$, $M'$ has at least $|E|+1$ edges.
This is a clear contradiction with $M'$ being a minimum sized cover, as any light cover of $G$ is a regular cover of $G'$, and $G$ only has $|E|$ edges in total.
\end{proof}

%

\section{Hardness}\label{sec:hard}

Previously, \cite{frv-mvdha-18} gave an approximation-preserving reduction from Vertex Cover to both \mvd and \mvid.  Thus, both are APX-complete, and in particular are hard to approximate within a factor of $2-\eps$ for any $\eps>0$, assuming UGC \cite{Khot02a}. Since these hardness results were proven for complete graphs, they also immediately apply to the general problems \gmvd and \gmvid. In this section we give stronger hardness results for \gmvid and \gmvd by giving approximation-preserving reductions from \mcut and \lcut.

\begin{problem}[\mcut]
Given an undirected unweighted graph $G=(V,E)$ on $n=|V|$ vertices together with $k$ pairs of vertices $\{s_i , t_i \}^k_{i=1}$, compute a minimum size subset of edges $M\subseteq E$ whose removal disconnects all the demand pairs, i.e., in the subgraph $(V,E \setminus M)$ every $s_i$ is disconnected from its corresponding vertex $t_i$.
\end{problem} 
\cite{ChawlaKKRS06} proved that if UGC is true, then it is NP-hard to approximate \mcut within any constant factor $L>0$, and assuming a stronger version of UGC, within $\Omega(\sqrt{\log\log n})$.
(The \mcut version in \cite{ChawlaKKRS06} allowed weights, but they remark their hardness proofs extend to the unweighted case.) 

\begin{theorem}\label{thm:mcut_to_gmvid}
 There is an approximation-preserving, polynomial-time reduction from \mcut to \gmvid.
\end{theorem}
\begin{proof}
Let $G=(V,E)$ be an instance of \mcut with $k$ pairs of vertices $\{s_i , t_i \}^k_{i=1}$. First, if $(s_i,t_i)\in E$ for any $i$, then that edge must be included in the solution $M$.  
Thus, we can assume no such edges exists in the \mcut instance, as assuming this can only make it harder to approximate the optimum value of the \mcut instance. 
We now construct an instance of \gmvid, $G'=(V',E',w)$. Let $V'=V$ and $E'=E\cup \{s_i , t_i \}^k_{i=1}$ where the edges in $E$ have weight one and the edges $(s_i, t_i)$, for all $i\in [k]$, have weight $n=|V|$.

Observe that if a cycle in $G'$ has exactly one edge of weight $n$, then it must be broken since there can be at most $n-1$ other edges in the cycle. Conversely, if a cycle $C$ has no edge with weight $n$ or more than one edge with weight $n$, then $C$ does not have a heavy edge, and so is not broken. 

Note that the edges from $G$ are exactly the weight one edges in $G'$, and thus, the paths in $G$ are in one-to-one correspondence with the paths in $G'$ which consist of only weight one edges. 
Moreover, the weight $n$ edges in $G'$ are in one-to-correspondence with the $(s_i,t_i)$ pairs from $G$. 
Thus, the cycles in $G'$ with exactly one weight $n$ edge followed by paths of all weight one edges connecting their endpoints, which by the above are exactly the set of broken cycles, are in one-to-one correspondence with paths between $(s_i,t_i)$ pairs from $G$.
Therefore, a minimum cardinality subset of edges which light cover all broken cycles, i.e., an optimal \gmvid support, corresponds to a minimum cardinality subset of edges from $E$ which cover all paths from $s_i$ to $t_i$ for all $i$, i.e., an optimal solution to \mcut. 
\end{proof}

\begin{problem}[\lcut]
Given a value $L$ and an undirected unweighted graph $G=(V,E)$ with source $s$ and sink $t$, find a minimum size subset of edges $M\subseteq E$ such that no $s$-$t$-path of length less than or equal to $L$ remains in the graph after removing the edges in $M$.
\end{problem}
An instance of \lcut with length $L$, is referred to as an instance of $L$-\lcut.
For any fixed $L$, Lee~\cite{l-ihcifp-17} showed that it is hard to approximate $L$-\lcut within a factor of $\Omega(\sqrt{L})$. 

\begin{theorem}
\label{thm:lbcut_to_gmvid}
 For any fixed value $L$, there is an approximation-preserving, polynomial-time reduction from $L$-\lcut to \gmvid.
\end{theorem}

\begin{proof}
Let $G=(V,E)$ be an instance of $L$-\lcut with source $s$ and sink $t$.
First, if $(s,t)\in E$, then that edge must be included in the solution $M$.  
Thus we can assume that edge is not in the \lcut instance, as assuming this can only make it harder to approximate the optimum value of the \lcut instance. 
We now construct an instance of \gmvid, $G'=(V',E',w)$. Let $V'=V$ and $E'=E\cup \{(s , t)\}$ where the edges in $E$ have weight $1$ and the edge $(s, t)$ has weight $L+1$.

First, observe that any cycle containing the edge $(s,t)$ followed by $\leq L$ unit weight edges is broken, as the sum of the unit weight edges will be $<L+1=w((s,t))$.
Conversely, any broken cycle must contain the edge $(s,t)$ followed by $\leq L$ unit weight edges.  
Specifically, if a cycle does not contain $(s, t)$ then it is unbroken since all edges would then have weight $1$.
Moreover, if a cycle contains $(s,t)$ and $>L$ other edges, then the total sum of those unit edges will be $\geq L+1 = w((s,t))$.

Note that the edges from $G$ are exactly the weight one edges in $G'$, and thus the paths in $G$ are in one-to-one correspondence with the paths in $G'$ which consist of only weight one edges. 
Moreover, the edge $(s, t)$ in $G'$ corresponds with the source and sink from $G$. 
Thus by the above, the broken cycles in $G'$ are in one-to-one correspondence with $s$-$t$-paths with length $\leq L$ in $G$.
Therefore, a minimum cardinality subset of edges which light cover all broken cycles, i.e., an optimal support to \gmvid, 
corresponds to a minimum cardinality subset of edges from $E$ which cover all paths from $s$ to $t$ of length $\leq L$, i.e., an optimal solution to \lcut. 
\end{proof}

In both the reductions from \gmvid to \gmvd of Theorem \ref{thm:gmvid_to_gmvd} and from $L$-\lcut to \gmvid of Theorem \ref{thm:lbcut_to_gmvid}, the maximum edge weight increases by $1$. Moreover, in the reduction from $L$-\lcut to \gmvid all but one edge (the $s,t$ pair) has unit weight.
%
%
Thus, based on these reductions, and previous hardness results, we have the following summarizing theorem.

\begin{theorem}\label{thm:last_hardness}
\gmvid and \gmvd are APX-complete, and moreover assuming UGC neither can be approximated within any constant factor. 

For any positive integer $L$, consider the problem defined by the restriction of \gmvd to integer weight instances with maximum edge weight $L$ and minimum edge weight 1, or the further restriction of \gmvid to instances where all weights are $1$ except for a single weight $L$ edge. 
Then assuming UGC these problems are hard to approximate within a factor of $\Omega(\sqrt{L})$.
\end{theorem}



\section{Fixed Parameter Analysis for $\varsigma$-Chordal Graphs}
\label{sec:fpt}

Throughout, let $\varsigma$ be a fixed constant, and let $F_\varsigma$ denote the family of all $\varsigma$-chordal graphs.
In this section we provide an FPT algorithm for MR$(G,\mathbb{R})$ for any $G\in F_\varsigma$. 

By Theorem \ref{thm:structure}, we seek a minimum sized cover of all broken cycles. First, we argue below that if $G$ has a broken cycle, then it has a broken chordless cycle. This seems to imply a natural FPT algorithm for constant $\varsigma$.  Namely, find an uncovered broken chordless cycle and recursively try adding each one of its edges to our current solution.\footnote{Indeed, one might be tempted to construe this algorithmic approach as kernelization, as in typical FPT algorithms. The edges of the broken chordless cycles do form a kernel but not one whose size is bounded in our parameter. 
As a simple example to illustrate this phenomenom, take $G = K_n$, set one edge weight to $n+1$, and set all other edge weights to 1. There are $2n-3$ edges in the kernel while the optimal solution has size 1.} However, it is possible to cover all broken chordless cycles while not covering the chorded cycles.  These cycles are difficult to cover as they may be much larger than $\varsigma$, though again by Theorem \ref{thm:structure} they must be covered.

Consider an optimal solution $W$, with support $S_W$. Suppose that we have found a subset $S\subsetneq S_W$, covering all broken chordless cycles in $G$. 
Intuitively, if we add to each edge in $S$ its weight from $W$, then any remaining broken chordless cycle must be covered further, in effect revealing which edges to consider from the chorded cycles from the original graph $G$. The challenge, however, is of course that we don't know $W$ a priori. We argue that despite this one can still identify a bounded sized subset of edges containing an edge from a cycle needing to be covered further.

\begin{lemma} \label{lem:chordal} 
If $G$ has a broken cycle, then $G$ has a broken chordless cycle.
\end{lemma}
\begin{proof} Let $C = v_1, \hdots, v_k$ be the broken cycle in $G$ with the fewest edges, with $v_1v_k$ being the heavy edge. If $C$ is chordless, then the claim holds. 
Otherwise, this cycle has at least one chord $v_iv_j$. Now there are two paths $P_1$ and $P_2$ from $v_i$ to $v_j$ on the cycle. Let $P_1$ be the path containing the heavy edge of $C$. 
If $w(v_i, v_j) >  \sum_{e \in P_2} w(e)$, then $P_2$ together with the edge $v_iv_j$ defines a broken cycle with fewer edges than $C$. On the other hand, if $w(v_i, v_j) \leq  \sum_{e \in P_2} w(e)$ then $P_1$ together with the edge $v_iv_j$ defines a broken cycle with fewer edges than $C$.
In either case we get a contradiction as $C$ was the broken cycle with the fewest edges.
\end{proof}

Our FPT algorithm is shown in Algorithm \ref{alg:fptrecurse}. The following lemma is key to arguing correctness.

\begin{algorithm}[ht]
\caption{FPT}
\label{alg:fptrecurse}
\begin{algorithmic}[1]
\Function{F}{$G,S,k$}
	\If{$|S| = k$} \Return \Call{verifier}{$G$,$S$} \label{line:returnverifier} \EndIf
		
	\State $P = \emptyset$
	\If{there exists a broken chordless cycle $C$ such that $C \cap S = \emptyset$}  \label{line:disjoint}
		\State  $P=C$  \label{line:disjoint2}
	\Else
		\For{$s \subseteq S$ such that $|s| \le \varsigma - 1$}
			\State Let $\mathcal{C} = \{ \text{Chordless cycles } C \text{ such that }  C \cap S = s \}$ \label{line:cset}
			\State $C_1  \leftarrow \arg\min_{C \in \mathcal{C}} \sum_{e \in C \setminus s} w(e)$ \label{line:c1}
			\State $C_2  \leftarrow \arg\max_{C \in \mathcal{C}} w(h) - \sum_{e \in C \setminus (s \cup \{h\})}$, where $h$ is the max weight edge in $C \setminus s$ \label{line:c2}
			\State Add $(C_1 \cup C_2) \setminus S$ to $P$ 
		\EndFor
	\EndIf
	\For{$e  \in P$}
		\State $X$ = \Call{F}{$G,S \cup \{e\},k$}
		\If{$X \neq$ NULL} \Return $X$ \EndIf \EndFor
	\State \Return NULL
\EndFunction

$ $

\Function{FPTWrapper}{$G$}
	\For{$k = 1, 2, \hdots$}
		\State $X$ = \Call{F}{$G,\emptyset, k$}
		\If{$X \neq$ NULL} \Return $X$ \EndIf \EndFor
\EndFunction
\end{algorithmic}
\end{algorithm}

\begin{lemma} \label{lem:fptrecurse} Consider any optimal solution $W$ and its support $S_W$ to an instance of metric repair for $G = (V,E,w)\in F_\varsigma$. If $S \subsetneq S_W$, then $F(G, S, \opt)$ adds at least one edge in $S_W \setminus S$ to $P$. \end{lemma}
\begin{proof} Consider the auxiliary graph $G_S = (V,E,\tilde{w})$, which has the same vertex and edge sets as $G$, but with the modified weight function: \[ \tilde{w} = \begin{cases} w(e) & e \not\in S \\ W(e) + w(e) & e \in S \end{cases} \]
Since $S \subsetneq S_W$, we have that $G_S$ has a broken cycle. Thus, by Lemma \ref{lem:chordal}, $G_S$ has a chordless broken cycle.
Suppose there is a chordless broken cycle in $G_S$ that is edge disjoint from $S$ (which occurs if and only if it is also broken in $G$), in which case, line \ref{line:disjoint} finds such a cycle. As this is a broken cycle, it must be covered by some edge in $S_W\setminus S$, and thus, we have added some edge in $S_W \setminus S$ to $P$.

Let us assume otherwise, that any chordless broken cycle in $G_S$ has non-empty intersection with $S$.  Let $C$ be any such chordless broken cycle with $C \cap S \neq \emptyset$.
Observe that as $C$ is broken in $G_S$, it must be that $|C\cap S|<|C|$, as otherwise it would imply $W$ was not a solution.
Thus, as $G \in F_\varsigma$, we know that $|C| \le \varsigma$, and so  $|C \cap S| < \varsigma$.
This implies in some for loop iteration, $C\in \mathcal{C}$ on line \ref{line:cset}.

Let $h$ be the heavy edge, in $G_S$, of the broken cycle $C$. We now have two cases:

\textbf{Case 1:} $h \in S$. In this case we have that \[ W(h) + w(h) > \underbrace{\sum_{e \in C \setminus S} w(e) }_{(1)}+ \sum_{e \in S} W(e) + w(e). \] On line \ref{line:c1} we found a cycle $C_1$ that minimized (1). Thus, since $C$ is broken in $G_S$, $C_1$ is also broken in $G_S$, and so must be covered by some edge in $S_W\setminus S$. Hence, we added some edge in $S_W \setminus S$ to $P$. 

\textbf{Case 2} $h \not\in S$. In this case $h$ has the maximum weight of all edges in $C\setminus s$. We have that  \[ \underbrace{w(h) - \sum_{e \in C \setminus (S \cup \{h\})} w(e)}_{(2)} > \sum_{e \in S} W(e) + w(e). \] On line \ref{line:c2} we found a cycle $C_2$ maximizing (2). Thus, if $C$ is broken in $G_S$, then $C_2$ is broken in $G_S$, and so must be covered by some edge in $S_W\setminus S$. Hence, we added some edge in $S_W \setminus S$ to $P$. 
%
\end{proof}

\begin{lemma} \label{lem:bound}Any time we call $F$, we have that $|P| \le 2\varsigma|S|^{\varsigma}$ 
\end{lemma}
\begin{proof} Note $|P|$ is upper bounded by $\varsigma$ multiplied by the number of chordless cycles we add. If the conditional on line \ref{line:disjoint} is true then we add only a single chordless cycle to $P$. Otherwise, for each $s \subseteq S$ such that $|s| \le \varsigma - 1$ we find two cycles. 
There are at most
\[ \sum_{i=1}^{\varsigma - 1} \binom{|S|}{i}\leq \sum_{i=1}^{\varsigma - 1} |S|^{i} \le |S|^{\varsigma}\] many such subsets, and thus we add at most $2|S|^{\varsigma}$ many cycles, implying the claim. 
\end{proof}


\begin{theorem}For any fixed constant $\varsigma$, Algorithm \ref{alg:fptrecurse} is an FPT algorithm for MR$(G,\mathbb{R})$ for any $G \in F_\varsigma$, when parameterized by \opt. The running time is $\Theta((2\varsigma\opt^{\varsigma})^{\opt+1} n^\varsigma)$. \end{theorem}
\begin{proof} 
{\sc FPTWrapper} iteratively calls $F(G,\emptyset, k)$ for increasing values of $k$ until it returns a non-{\sc Null} value. 
First, we argue that while $k<\opt$, $F(G,\emptyset, k)$ will return {\sc Null}. In the initial call to $F$, we have $S=\emptyset$. $F$ then adds exactly one edge in each recursive call until $|S|=k$, at which point it returns {\sc Verifier}$(G,S)$. Thus, as $k<\opt$, by  proposition \ref{prop:verifier}, {\sc NULL} is returned.


Now we argue that when $k = \opt$ an optimal solution is returned. 
Fix any optimal solution $W$ and its support $S_W$ to the given instance $G$. By Lemma \ref{lem:fptrecurse}, if $S\subsetneq S_W$ (which is true initially as $S=\emptyset$) then at least one edge in $S_W\setminus S$ is added to $P$. Thus, as $F$ makes a recursive call to $F(G,S\cup\{e\},k)$ for every edge $e\in P$, in at least one recursive call an edge of $S_W$ is added to $S$. Thus there is some path in the tree of recursive calls to $F$ in which all $k=OPT$ edges from $S_W$ are added, at which point $F$ returns {\sc Verifier}$(G,S)$, which returns an optimal solution by proposition \ref{prop:verifier}. (Note this recursive call may not be reached, if a different optimal solution is found first.)


Now we consider bounding the running time. Observe that in each call to $F$, a set $P$ is constructed, and then recursive calls to $F(G,S\cup\{e\},k)$ are made for each $e\in P$. By Lemma \ref{lem:bound}, $|P| \le 2\varsigma|S|^{\varsigma}  \le 2\varsigma k^{\varsigma}$ at all times. So in the tree of all recursive calls made by any initial call to $F(G,\emptyset,k)$, the branching factor is always bounded by $2\varsigma k^{\varsigma}$, and the depth is $k$. 
Thus there are $O((2\varsigma k^{\varsigma})^{k})$ nodes in our recursion tree. 

Now we bound the time needed for each node in the recursion tree. If {\sc Verifier} is called then it takes $O(n^3)$ time by proposition \ref{prop:verifier}. Otherwise, note that there are $O(n^{\varsigma})$ chordless cycles. Thus it takes $O(\varsigma n^{\varsigma})$ time to enumerate and check them on line \ref{line:disjoint}. Similarly $|\mathcal{C}|=O(n^{\varsigma})$ on line \ref{line:cset}, and so  the run time of each iteration of the for loop is $O(\varsigma n^{\varsigma})$. There are $O(|S|^{\varsigma}) = O(k^{\varsigma})$ iterations of the for loop, thus the total time per node is $O(\varsigma k^{\varsigma} n^{\varsigma})$.

Thus the total time for each call to $F(G,\emptyset,k)$ is $O((2\varsigma k^{\varsigma})^{k} \varsigma k^{\varsigma} n^{\varsigma}) = O((2\varsigma k^{\varsigma})^{k+1} n^{\varsigma})$.
Since {\sc FPTWrapper} calls $F(G,\emptyset,k)$ for $k=1, \hdots, \opt$, the overall running time of our algorithm is 
\[ O\left( \left( \sum_{k=1}^{\opt} (2\varsigma k^{\varsigma})^{k+1} \right) \cdot n^\varsigma \right) = O((2\varsigma \opt^{\varsigma})^{\opt+1} n^\varsigma) \qedhere\] 
\end{proof}

As lemma \ref{lem:fptrecurse} holds for any optimal solution, the bound on the recursion tree size in the above proof actually bounds the number of optimal solutions.

\begin{corollary} If $G \in F_\varsigma$ then there are at most $(2\varsigma\opt^{\varsigma})^{\opt}$ subsets $S \subset E$ such that $S$ is the support of an optimal solution to MR$(G,\mathbb{R})$. 
\end{corollary}

\begin{remark}
Using the approximation-preserving, polynomial-time reduction from \gmvid to \gmvd in Theorem \ref{thm:gmvid_to_gmvd}, the above also yields an FPT for \gmvid.  This holds since the reduction does not change the optimal solution size, nor does it change $\varsigma$ as it only adds triangles.  
Alternatively, the above algorithm can be carefully modified to consider light covering broken cycles. 
\end{remark}






\section{Approximation Algorithms}\label{sec:approx}

In this section we present approximation algorithms for \gmvid and \gmvd. 

By Theorem \ref{thm:structure}, we know the support of an optimal solution to \gmvd is a minimum cardinality regular cover of all broken cycles.  
This naturally defines a hitting set instance $(E,\mathcal{C})$, where the ground set $E$ is the edges from $G$, and $\mathcal{C}$ is the collection of the subsets of edges determined by the broken cycles. 
Unfortunately, constructing $(E,\mathcal{C})$ explicitly is infeasible as there may be an exponential number of broken cycles. 
In general just counting the number of paths in a graph is \#P-Hard \cite{v-cerp-79}, 
though it is known how to count paths of length up to roughly $O(\log n)$ using color-coding. 
(See \cite{ag-bfphfa-10} and references therein.  Also \cite{bdh-ec-18} for recent FPT algorithms.) 
Moreover, observe our situation is more convoluted as we wish to count only paths corresponding to broken cycles. 
                                        

Despite these challenges, we argue there is sufficient structure to at least roughly apply the standard greedy algorithms for hitting set.
Our first key insight, related to insight~\ref{item:broken}, is:
\begin{enumerate}[resume*=insights]
\item One can always find \emph{some} broken cycle, if one exists, by finding any edge whose weight is more than the shortest path length between its endpoints (using APSP). \label{item:set}
\end{enumerate}
In the language of hitting set, we have a polynomial time oracle, which returns an arbitrary set in $\mathcal{C}$. 
Recall the simple greedy algorithm for hitting set, which repeatedly picks an arbitrary uncovered set, and adds all its elements to the solution.
If $L=\max_{c \in \mathcal{C}}|c|$ denotes the largest set size, then this algorithm gives an $L$-approximation, as each time we take the elements of a set, we get at least one element from the optimal solution.  
Below we apply this approach to approximate \gmvd and \gmvid.

We would prefer, however, to have an oracle for the number of broken cycles that an edge $e\in E$ participates in as using such an oracle would yield an $O(\log n)$-approximation algorithm for \gmvd (regardless of the size of $L$) by running the standard greedy algorithm for hitting set which repeatedly selects the element that hits the largest number of uncovered sets.  Towards this end, we have the following key insight:
\begin{enumerate}[resume*=insights]
\item We can find the \emph{most} broken cycle (i.e., with maximum deficit) and, more importantly, count how many such maximum deficit cycles each edge is in. \label{item:deficit}
\end{enumerate}
To argue that insight \ref{item:deficit} is true, first we observe that the cycle with the largest deficit value corresponds to a shortest path. This in turn, we argue over several lemmas, allows us to quickly get a count when restricting to such cycles. 
Thus, if $\deficit$ denotes the number of distinct cycle deficit values, we can show that the above insight implies an $O(\deficit \log n)$-approximation, by breaking the problem into $\deficit$ instances of hitting set, where for each instance we can run the greedy algorithm.

\subsection{$L$-approximation}

In this section, we consider the problems defined by restricting \gmvd and \gmvid to the subset of instances where the largest number of light edges in a broken cycle is $L$. 
We present an $(L+1)$-approximation algorithm for \gmvd which runs in $O(n^3\cdot \opt)$ time, which also will imply an $L$-approximation for \gmvid with the same running time.

As mentioned above, the main idea comes from insight \ref{item:set}. 
In particular, the following algorithm, {\sc Short Path Cover (SPC)}, 
can be easily understood by viewing it as running the standard $L$-approximation for the corresponding instance $(E,\mathcal{C})$ of hitting set, where we have an oracle for finding a set $c \in \mathcal{C}$.
In the following, $\mathsf{APSP}$ is a subroutine returning a shortest path distance function $d(u,v)$, and a function $P(u,v)$ giving the set of edges along \emph{any} shortest path from $u$ to $v$.

\begin{algorithm}
\caption{Short Path Cover (SPC) for \gmvd}
\label{alg:lapx}
\begin{algorithmic}[1]
\Function{SPC}{$G=(V,E,w)$}
\State $H = (V_H=V, E_H=E, w_H=w)$
\While{True}
\State $d,P = \mathsf{APSP}(H)$ \label{line_apsp}
\If{$\exists\ e=(u,v)\in E_H$ such that $w(e) > d(u,v)$} \label{exists_edge}
\State $E_H = E_H \setminus \left ( P(u,v) \cup \{e\} \right )$ \label{edge_union}
\Else{}
\State \Return \Call{Verifier}{$G, E\setminus E_H$} \label{rcall}
\EndIf
\EndWhile
\EndFunction
\end{algorithmic}
\end{algorithm}

\begin{theorem} \label{thm:spc} 
{\sc SPC} gives an $(L+1)$-approximation for \gmvd in $O(n^3\cdot \opt)$ time. 
\end{theorem}
\begin{proof}
First, note that if there is a broken cycle in $H$, then for some edge $e=(u,v)$, $w(e) > d(u,v)$, and moreover, in this case $P(u,v) \cup \{e\}$ is a broken cycle.
Thus, when the algorithm terminates there are no broken cycles in $H$. 
Also, for any broken cycle in $G$, if all of its edges are still in $H$, then it will be a broken cycle in $H$.
Thus, when the algorithm terminates at least one edge from each broken cycle in $G$ is in $E\setminus E_H$, 
which by Theorem \ref{thm:structure} implies $E\setminus E_H$ is a valid support.

Note that removing edges does not create any new broken cycles, thus, any broken cycle in $H$ is also a broken cycle in $G$.
Thus, the support of any optimum solution must contain at least one edge from each broken cycle in $H$ (again by Theorem~\ref{thm:structure}), 
and so every time we remove the edges of a broken cycle $P(u,v) \cup \{e\}$, we remove at least one optimum edge.
As the largest broken cycle length is $L+1$, this implies overall we get an $(L+1)$-approximation.
The same argument implies the while loop can get executed at most $\opt$ times, and as APSP takes $O(n^3)$ time via Floyd-Warshall, 
and line \ref{exists_edge} takes $O(n^2)$ time, we obtain the running time in the theorem statement.
\end{proof} 

\begin{remark}
If we modify {\sc SPC} so that in line \ref{edge_union} we only remove $P(u,v)$ from $E_H$ (rather than $P(u,v) \cup \{e\}$), 
then by the second part of Theorem~\ref{thm:structure}, the same argument implies that {\sc SPC} is an $L$-approximation for \gmvid that runs in $O(n^3\cdot \opt)$ time.
\end{remark}

\begin{remark}
 Theorem \ref{thm:last_hardness} restricts \gmvid and \gmvd to integer weight instances with max weight $L$, implying any broken cycle has at most $L$ edges.
 As this is a subset of the instances here,  
 {\sc SPC} is an $L$ or $L+1$ approximation for instances that are hard to approximate within $\Omega(\sqrt{L})$. 
\end{remark}

\subsection{$O(\deficit \log n)$-approximation}

Using insight \ref{item:deficit}, our approach is to iteratively cover cycles by decreasing deficit value, ultimately breaking the problem into multiple hitting set instances. We present the algorithm for \gmvd first and then remark on the minor change needed to apply it to \gmvid. 

%

\remove{
Let the \emph{deficit} of a cycle $C$, denoted $\delta(C)$, be equal to the weight of its heavy edge minus the sum of the weights of all other edges in $C$.
Similarly, let $\delta(G)$ denote the maximum value of $\delta(C)$ over all cycles.
Note the set of broken cycles is equivalently the set of cycles with strictly positive deficit.
}

For any pair of vertices $s,t\in V$, we write $d(s,t)$ to denote their shortest path distance in $G$, 
and $\csp(s,t)$ to denote the number of shortest paths from $s$ to $t$.
It is straightforward to show that $\csp(s,t)$ can be computed in $O(m+n)$ time given all $d(u,v)$ values have been precomputed. (See Lemma \ref{lem:pathCount} in the Appendix \ref{apnd:approx}.)

Recall that for a broken cycle $C$ with heavy edge $h$, the deficit of $C$ is $\delta(C) = w(h) - \sum_{e\in (C\setminus h)} w(e)$. 
Moreover, $\delta(G)$ denotes the maximum deficit over all cycles in $G$.
For any edge $e$, define $N_h(e,\alpha)$ to be the number of distinct broken cycles of deficit $\alpha$ whose heavy edge is $e$. 
Similarly, let $N_l(e,\alpha)$ denote the number of distinct broken cycles with deficit $\alpha$ which contain the edge $e$, but where $e$ is not the heavy edge.
While for general $\alpha$ it is not clear how to even approximate $N_l(e,\alpha)$ and $N_h(e,\alpha)$, we argue over several lemmas that when $\alpha=\delta(G)$ these values can be computed exactly.  

\begin{lemma}\label{lem:Topedge}
For any edge $e=(s,t)$, if $w(e)=d(s,t)+\delta(G)$ then $N_h(e,\delta(G))=\csp(s,t)$, and otherwise $N_h(e,\delta(G))=0$.
\end{lemma}
\begin{proof}
 If $w(e)\neq d(s,t)+\delta(G)$, then as $\delta(G)$ is the maximum deficit over all cycles, it must be that $w(e)< d(s,t)+\delta(G)$, 
 which in turn implies any broken cycle with heavy edge $e$ has deficit strictly less than $\delta(G)$.
 Now suppose $w(e)= d(s,t)+\delta(G)$, and consider any path $p_{s,t}$ from $s$ to $t$ such that $e$ together with $p_{s,t}$ creates a broken cycle with heavy edge $e$. 
 If $p_{s,t}$ is a shortest path then $w(e)-w(p_{s,t}) = w(e)-d(s,t)=\delta(G)$, and otherwise $w(p_{s,t})>d(s,t)$ and so $w(e)-w(p_{s,t})<w(e)-d(s,t)=\delta(G)$.
 Thus $N_h(e,\delta(G))=\csp(s,t)$ as claimed.
\end{proof}


As $G$ is undirected, every edge $e\in E$ correspond to some unordered pair $\{a,b\}$.  
However, often we write $e=(a,b)$ as an ordered pair, 
according to some fixed arbitrary total ordering of all the vertices. 
We point this out to clarify the following statement.
\begin{lemma}\label{lem:Nonheavyedge}
Fix any edge $e=(s,t)$, and let 
$X=\{f=(a,b) \mid w(f)=d(a,s)+w(e)+d(t,b)+\delta(G) \}$, and 
$Y=\{f=(a,b) \mid w(f)=d(b,s)+w(e)+d(t,a)+\delta(G) \}$.
Then it holds that
\[
N_l(e,\delta(G))=\left(\sum_{(a,b)\in X} \csp(a,s)\cdot \csp(t,b)\right)+\left(\sum_{(a,b)\in Y} \csp(b,s)\cdot \csp(t,a)\right).
\]
\end{lemma}
\begin{proof}
Consider any broken cycle $C$ containing $e=(s,t)$, with heavy edge $f=(a,b)$ and where $\delta(C)=\delta(G)$.
Such a cycle must contain a shortest path between $a$ and $b$, as otherwise it would imply $\delta(G)>\delta(C)$.
Now if we order the vertices cyclically, then the subset of $C$'s vertices $\{a,b,s,t\}$, must appear either in the order $a,s,t,b$ or $b,s,t,a$.
In the former case, as the cycle must use shortest paths, $w(f)=d(a,s)+w(e)+d(t,b)+\delta(G)$, and the number of cycles satisfying this is $\csp(a,s)\cdot \csp(t,b)$.  
In the latter case, $w(f)=d(b,s)+w(e)+d(t,a)+\delta(G)$, and the number of cycles satisfying this is $\csp(b,s)\cdot \csp(t,a)$.
Note also that the set $X$ from the lemma statement is the set of all $f=(a,b)$ satisfying the equation in the former direction, and $Y$ is the set of all $f=(a,b)$ satisfying the equation in the later direction.
Thus summing over each relevant heavy edge in $X$ and $Y$, of the number of broken cycles of deficit $\delta(G)$ which involve that heavy edge and $e$, yields the total number of broken cycles with deficit $\delta(G)$ 
containing $e$ as a light edge.
\end{proof}

\begin{corollary}\label{cor:time} Given constant time access to $d(u,v)$ and $\csp(u,v)$ for any pair of vertices $u$ and $v$, $N_h(e,\delta(G))$ can be computed in $O(1)$ time and $N_l(e,\delta(G))$ in $O(m)$ time.
\end{corollary}
\begin{proof}
By Lemma \ref{lem:Topedge}, in constant time we can check whether $w(e)=d(s,t)+\delta(G)$, in which case set $N_h(e,\delta(G))=\csp(s,t)$, and otherwise set $N_h(e,\delta(G))=0$. 
By Lemma \ref{lem:Nonheavyedge}, we can compute $N_l(e,\delta(G))$ with a linear scan of the edges, 
where for each edge $f$ in constant time we can compute whether $w(f)=d(a,s)+w(e)+d(t,b)+\delta(G)$ and if so add $\csp(a,s)\cdot \csp(t,b)$ to the sum over $X$, 
and if $w(f)=d(b,s)+w(e)+d(t,a)+\delta(G)$ add $\csp(b,s)\cdot \csp(t,a)$ to the sum over $Y$.
\end{proof}


\begin{algorithm}[ht]
\caption{Finds a valid solution for \gmd.}
\label{alg:simp}
\begin{algorithmic}[1]
\Function{Approx}{$G=(V,E,w)$}
\State    Let $S=\emptyset$\;
   \While{True}
    \State For every pair $s,t\in V$ compute $d(s,t)$\;\label{allPairs}
    \State Compute $\delta(G) = \max_{e=(s,t)\in E} ~w(e)-d(s,t)$\;\label{deficits}
    \If{$\delta(G) = 0$}
       \Return \Call{Verifier}{$G, S$}\;\label{returnVerifier}
    \EndIf
    \State For every edge $(s,t)\in E$ compute $\csp(s,t)$\;\label{pathCounting}
    \State For every $e\in E$ compute $count(e) = N_h(e,\delta(G))+N_l(e,\delta(G))$\;\label{count} 
    \State Set $f = \arg\max_{e\in E} count(e)$\;\label{maxF}
    \State Update $S=S\cup\{f\}$ and $G=G\setminus f$\;\label{edgeRemove}
   \EndWhile
\EndFunction
\end{algorithmic}
\end{algorithm}

\begin{theorem}\label{thm:deficitMain}
 For any positive integer $\deficit$, consider the set of \gmvd instances where the number of distict deficit values is at most $\deficit$, i.e., $|\{\delta(C)\mid \text{$C$ is a cycle in $G$} \}|\leq \deficit$.
 Then Algorithm \ref{alg:simp} gives an $O((n^3+m^2)\cdot \opt\cdot \deficit\log n)$ time $O(\deficit\log n)$-approximation. 
\end{theorem}
\begin{proof}
 Observe that the algorithm terminates only when $\delta(G) = 0$, i.e., only once there are no broken cycles left. 
 As no new edges are added, and weights are never modified, this implies that when the algorithm terminates it outputs a valid regular cover $S$.
 (The algorithm must terminate as every round removes an edge.)  
 Therefore, by Theorem \ref{thm:structure}, $S$ is a valid \gmvd support, and so we only need to bound its size.

 Let the edges in $S=\{s_1,\ldots,s_k\}$ be indexed in increasing order of the loop iteration in which they were selected.
 Let $G_1,\ldots, G_{k+1}$ be the corresponding sequence of graphs produced by the algorithm, where $G_i=G\setminus \{s_1,\ldots,s_{i-1}\}$.
 Note that for all $i$, $G_i=(V,E_i)$ 
 induces a corresponding instance of hitting set, $(E_i,\mathcal{C}_i)$, 
 where the ground set is the set of edges from the \gmvd instance $G_i$, and $\mathcal{C}_i = \{E_i(C)\mid \text{$C$ is a broken cycle in $G_i$}\}$ (where $E_i(C)$ is the set of edges in $C$). 
 
 Let $D=\{\delta(C)\mid \text{$C$ is a cycle in $G$} \}$, where by assumption $|D|\leq \deficit$.
 Note that any cycle $C$ in any graph $G_i$, is also a cycle in $G$. 
 Thus as we never modify edge weights, $\delta(G_1),\ldots,\delta(G_{k+1})$ is a non-increasing sequence. Moreover $X=\{\delta(G_i)\}_i\subseteq D$, and in particular $|X|\leq \deficit$.
 For a given value $\delta\in X$, let $G_{\alpha},G_{\alpha+1},\ldots,G_{\beta}$ be the subsequence of graphs with deficit $\delta$ (which is consecutive as the deficit values are non-increasing).
 Observe that for all $\alpha\leq i\leq \beta$, the edge $s_i$ is an edge from a cycle with deficit $\delta=\delta(G_i)$.  
 So for each $\alpha\leq i\leq \beta$, define a sub-instance of hitting set $(E_i', \mathcal{C}_i')$, where $E_i'$ is the set of edges in cycles of deficit $\delta$ from $G_i$, 
 and $\mathcal{C}_i'$ is the family of sets of edges from each cycle of deficit $\delta$ in $G_i$.
 
 The claim is that for the hitting set instance $(E_\alpha', \mathcal{C}_\alpha')$, that $\{s_\alpha,\ldots,s_\beta\}$ is an $O(\log n)$ approximation to the optimal solution. 
 To see this, observe that for any $\alpha\leq i\leq \beta$ in line \ref{count}, $count(e)$ is the number of times $e$ is contained in a broken cycle with deficit $\delta =\delta(G_i)$, 
 as by definition $N_h(e,\delta(G_i))$ and $N_l(e,\delta(G_i))$ count the occurrences of $e$ in such cycles as a heavy edge or light edge, respectively.
 Thus $s_i$ is the edge in $E_i'$ which hits the largest number of sets in $\mathcal{C}_i'$, 
 and moreover, $(E_{i+1}', \mathcal{C}_{i+1}')$ is the corresponding hitting set instance induced by removing $s_i$ and the sets it hit from $(E_i', \mathcal{C}_i')$. 
 Thus $\{s_\alpha,\ldots,s_\beta\}$ is the resulting output of running the standard greedy hitting set algorithm on $(E_\alpha', \mathcal{C}_\alpha')$ 
 (that repeatedly removes the element hitting the largest number of sets), and it is well known this greedy algorithm produces an $O(\log n)$ approximation.

 The bound on the size of $S$ now easily follows.  Specifically, let $I=\{i_1, i_2,\ldots,i_{|X|}\}$ be the collection of indices, where $i_j$ was the first graph considered with deficit $\delta(G_{i_j})$.  
 By the above, $S$ is the union of the $O(\log n)$-approximations to the sequence of hitting set instance $(E_{i_1}', \mathcal{C}_{i_1}'),\ldots, (E_{i_{|X|}}', \mathcal{C}_{i_{|X|}}')$. 
 In particular, note that for all $i_j$, $(E_{i_j}', \mathcal{C}_{i_j}')$ is a hitting set instance induced from the removal of a subset of edges from the initial hitting set instance $(E_1, \mathcal{C}_1)$, 
 and then further restricted to sets from cycles with a given deficit value.
 Thus the size of the optimal solution on each of these instances can only be smaller than on $(E_1, \mathcal{C}_1)$. 
 This implies that the total size of the returned set $S$ is $O(\opt\cdot |X|\log n)=O(\opt\cdot \deficit\log n)$.
 
 As for the running time, first observe that by the above, there are $O(\opt\cdot \deficit\log n)$ while loop iterations. Next, the single call to {\sc Verifier} in line \ref{returnVerifier} takes $O(n^3)$. For a given loop iteration, computing all pairwise distance in line \ref{allPairs} also takes $O(n^3)$ time using the standard Floyd-Warshall algorithm. Computing the graph deficit in line \ref{deficits} can then be done in $O(m)$ time. For any given vertex pair $s,t$, computing $\csp(s,t)$ takes $O(m+n)$ time by Lemma \ref{lem:pathCount}. Thus computing the number of shortest paths over all edges in line \ref{pathCounting} takes $O(m^2+mn)$ time.  
 For each edge $e$, by Corollary \ref{cor:time}, $count(e) = N_h(e,\delta(G)) + N_l(e,\delta(G))$ can be computed in $O(m)$ time, and thus computing all counts in line \ref{count} takes $O(m^2)$ time.
 As the remaining steps can be computed in linear time, each while loop iteration in total takes $O(n^3+mn+ m^2) = O(n^3+m^2)$ time, thus implying the running time bound over all iterations in the theorem statement. 
\end{proof}

\begin{remark}
 Rather than computing the $d(u,v)$ values from scratch in each iteration, we can use a dynamic data structure.
 This would slightly improve the time to $O(n^3+(n^{2+\alpha}+m^2)\cdot \opt\cdot \deficit\log n)$, where $0\leq \alpha$ 
 is a constant depending on the query and update time of the dynamic data structure.
 (Ignoring $\log$ factors, $\alpha=3/4$ is known. See for example \cite{iss-fdapsp-17} and references therein).
 However, similarly improving the $m^2$ term is more challenging as the $N_l(e,\delta(G))$ values depend in a non-trivial 
 way on collections of $d(u,v)$ values, each of which may or may not have changed.
\end{remark}

\begin{remark}
 If we modify Algorithm \ref{alg:simp}, so that line \ref{count} instead sets $count(e) = N_l(e,\delta(G))$, then by Theorem~\ref{thm:structure}, 
 the same argument implies we have an algorithm with the same time and approximation factor for \gmvid.
 Alternatively, we could combine the above algorithm for \gmvd with the approximation-preserving reduction from \gmvid to \gmvd of Theorem \ref{thm:gmvid_to_gmvd}. 
 However, the reduction increases the graph size by a linear factor, resulting in a slower running time.
\end{remark}


\section{Conclusion}

In this paper we introduced and gave a number of results for the graph metric repair problem.  In particular, we gave three main results.  First, we reduced from the well know \mcut problem, suggesting a possible $\Omega(\log n)$ approximation lower bound. Also, we gave a reduction from $L$-\lcut, implying an $\Omega(\sqrt{L})$ approximation lower bound.
Next, we gave an FPT algorithm for the $\varsigma$-chordal case, answering an open question from previous work, and requiring significant structural insight into the problem. The natural next question is whether it can be extended to general graphs.
Finally, we gave an $L$-approximation and a non-trivial $O(\deficit \log n)$-approximation, which given the reduction from \mcut, begs the question of whether reducing or eliminating the dependence on the deficit $\deficit$ is possible. 

\bibliographystyle{plain}
\bibliography{citations} 

\begin{thebibliography}{10}

\bibitem{iss-fdapsp-17}
I.~Abraham, S.~Chechik, and S.~Krinninger.
\newblock Fully dynamic all-pairs shortest paths with worst-case update-time
  revisited.
\newblock In {\em Proceedings of the Twenty-Eighth Annual {ACM-SIAM} Symposium
  on Discrete Algorithms (SODA)}, pages 440--452, 2017.

\bibitem{AbrahamGGNT14}
I.~Abraham, C.~Gavoille, A.~Gupta, O.~Neiman, and K.~Talwar.
\newblock Cops, robbers, and threatening skeletons: padded decomposition for
  minor-free graphs.
\newblock In {\em Symposium on Theory of Computing (STOC)}, pages 79--88, 2014.

\bibitem{ag-bfphfa-10}
N.~Alon and S.~Gutner.
\newblock Balanced families of perfect hash functions and their applications.
\newblock {\em {ACM} Trans. Algorithms}, 6(3):54:1--54:12, 2010.

\bibitem{Baraty2011}
Saaid Baraty, Dan~A. Simovici, and Catalin Zara.
\newblock The impact of triangular inequality violations on medoid-based
  clustering.
\newblock In Marzena Kryszkiewicz, Henryk Rybinski, Andrzej Skowron, and
  Zbigniew~W. Ra{\'{s}}, editors, {\em Foundations of Intelligent Systems},
  pages 280--289, Berlin, Heidelberg, 2011. Springer Berlin Heidelberg.

\bibitem{bdh-ec-18}
C.~Brand, H.~Dell, and T.~Husfeldt.
\newblock Extensor-coding.
\newblock In {\em Proceedings of the 50th Annual {ACM} {SIGACT} Symposium on
  Theory of Computing, {STOC} 2018, Los Angeles, CA, USA, June 25-29, 2018},
  pages 151--164, 2018.

\bibitem{Brickell}
Justin Brickell, Inderjit~S Dhillon, Suvrit Sra, and Joel~A Tropp.
\newblock The metric nearness problem.
\newblock {\em SIAM Journal on Matrix Analysis and Applications},
  30(1):375--396, 2008.

\bibitem{cr-emcco-12}
E.~Cand\`{e}s and B.~Recht.
\newblock Exact matrix completion via convex optimization.
\newblock {\em Commun. ACM}, 55(6):111--119, June 2012.

\bibitem{Chandran2005GraphsOL}
L.~Sunil Chandran, Vadim~V. Lozin, and C.~R. Subramanian.
\newblock Graphs of low chordality.
\newblock {\em Discrete Mathematics and Theoretical Computer Science},
  7:25--36, 2005.

\bibitem{ChawlaKKRS06}
S.~Chawla, R.~Krauthgamer, R.~Kumar, Y.~Rabani, and D.~Sivakumar.
\newblock On the hardness of approximating multicut and sparsest-cut.
\newblock {\em Computational Complexity}, 15(2):94--114, 2006.

\bibitem{frv-mvdha-18}
C.~Fan, B.~Raichel, and G.~{Van Buskirk}.
\newblock Metric violation distance: Hardness and approximation.
\newblock In {\em Proceedings of the Twenty-Ninth Annual {ACM-SIAM} Symposium
  on Discrete Algorithms (SODA)}, pages 196--209, 2018.

\bibitem{frv-mvdre-18}
C.~Fan, B.~Raichel, and G.~{Van Buskirk}.
\newblock Metric violation distance: Revisited and extended.
\newblock {\em CoRR}, abs/1807.08078, 2018.

\bibitem{GargVY96}
N.~Garg, V.~Vazirani, and M.~Yannakakis.
\newblock Approximate max-flow min-(multi)cut theorems and their applications.
\newblock {\em {SIAM} J. Comput.}, 25(2):235--251, 1996.

\bibitem{Gilbert2017}
A.~C. Gilbert and L.~Jain.
\newblock {If it ain't broke, don't fix it: Sparse metric repair}.
\newblock {\em ArXiv e-prints}, October 2017.

\bibitem{gs-gmrg-18}
A.~C. Gilbert and R.~Sonthalia.
\newblock Generalized metric repair on graphs.
\newblock {\em CoRR}, abs/1807.07619, 2018.

\bibitem{GilbertSonthalia:SwissRoll2018}
Anna~C. Gilbert and Rishi Sonthalia.
\newblock Unsupervised metric learning in presence of missing data.
\newblock In {\em 56th Annual Allerton Conference on Communication, Control,
  and Computing, Allerton 2018, Monticello, IL, USA, October 2-5, 2018}, pages
  313--321, 2018.

\bibitem{Khot02a}
S.~Khot.
\newblock On the power of unique 2-prover 1-round games.
\newblock In {\em Proceedings on 34th Annual {ACM} Symposium on Theory of
  Computing (STOC)}, pages 767--775, 2002.

\bibitem{l-ihcifp-17}
E.~Lee.
\newblock Improved hardness for cut, interdiction, and firefighter problems.
\newblock In {\em 44th International Colloquium on Automata, Languages, and
  Programming (ICALP)}, pages 92:1--92:14, 2017.

\bibitem{sww-meo-17}
A.~Sidiropoulos, D.~Wang, and Y.~Wang.
\newblock Metric embeddings with outliers.
\newblock In {\em Proc. Twenty-Eighth Annual {ACM-SIAM} Symposium on Discrete
  Algorithms (SODA)}, pages 670--689, 2017.

\bibitem{v-cerp-79}
L.~Valiant.
\newblock The complexity of enumeration and reliability problems.
\newblock {\em {SIAM} J. Comput.}, 8(3):410--421, 1979.

\bibitem{Wang2015}
Fei Wang and Jimeng Sun.
\newblock Survey on distance metric learning and dimensionality reduction in
  data mining.
\newblock {\em Data Mining and Knowledge Discovery}, 29(2):534--564, Mar 2015.

\end{thebibliography}

\appendix

\section{The Decrease Only Case}\label{apnd:decrease}

For the problem MR$(G,\mathbb{R}_{\le 0})$, consider the following simple algorithm, used in previous works for the special case when $G = K_n$.

\begin{algorithm}[ht]
\label{alg:DMR}
\caption{Decrease Metric Repair ({\sc Dmr})}
\begin{algorithmic}[1]
\Function{DMR}{$G = (V,E,w)$}
\State Let $W = w$
\State For any edge $e=uv \in E$, set $W(e)$ =  weight of a shortest path between $u$ and $v$ 
\State \Return $W-w$
\EndFunction
\end{algorithmic}
\end{algorithm}

{\reftheorem{thm:decreaseNew} 
The problem MR$(G, \mathbb{R}_{\le 0})$ can be solved in $O(n^3)$ time by the {\sc Dmr} algorithm.

Moreover, the problem becomes hard if even a single positive value is allowed.  That is, if $0 \in \Omega$ and $\Omega \cap \mathbb{R}_{> 0} \neq \emptyset$ then MR$(G, \Omega)$  is APX-Complete.
}
\begin{proof} For the first part, let $e \in G$ be an edge whose edge weight is bigger than the shortest path between the two end points of $e$. Then in this case $e$ is the heavy edge in a broken cycle. Hence, any decrease only solution must decrease this edge. Thus all edges decreased by {\sc Dmr} are edges that must be decreased.

By the same reasoning we see that this new weighted graph has no broken cycles. Thus, we see that our algorithm gives a sparsest solution to MR($G, \mathbb{R}_{\le 0})$ in $\Theta(n^3)$ time. 

For the second part, 
the reduction is the same as that of Fan et~al.\ \cite{frv-mvdha-18}. 
However, we make the observation that for any value $\alpha > 0$, by appropriately scaling the weights of the reduction in Fan et~al.\ \cite{frv-mvdha-18}, MR$(G, \mathbb{R}_{\le 0})$ is still APX-Hard in the extreme case when $\Omega = \{0,\alpha\}$.
\end{proof}

\begin{corollary} \label{dor:decp} 
For any $G = (V,E,w)$ {\sc Dmr} returns the smallest solution for any $\ell_p$ norm for $p \in [1, \infty)$. 
\end{corollary}
\begin{proof} The proof of Theorem \ref{thm:decreaseNew}  actually shows that there is a unique support for the sparsest solution, i.e., the set of all heavy edges. In fact any decrease only solution must contain all of these edges in its support. We can also see that {\sc Dmr} decreases these by the minimum amount so that the cycles are unbroken. Thus, this solution is in fact the smallest for any $\ell_p$ norm. 
\end{proof}

\section{Counting Shortest Paths}\label{apnd:approx}

%


\begin{lemma}\label{lem:pathCount}
Let $G$ be a positively weighted graph, where for all pairs of vertices $u,v$ one has constant time access to the value $d(u,v)$.
Then for any pair of vertices $s,t$, the value $\csp(s,t)$ can be computed in $O(m+n)$ time.
\end{lemma}
\begin{proof}
Let $V=\{v_1,v_2,v_3,...,v_n\}$, and let $N(v_i)$ denote the set of neighbors of $v_i$. 
Define  $X_i = \{ v_j\in N(v_i)\mid w(v_i,v_j)+d(v_j,t)=d(v_i,t)\}$, that is, $X_i$ is the set of neighbors of $v_i$ where there is a shortest path from $t$ to $v_i$ passing through that neighbor.
Thus we have,
\[
\csp(v_i,t)=\sum_{v_j\in X_i} \csp(v_j,t).
\]
Note that any shortest path from $v_i$ to $t$ can only use vertices $v_j$ which are closer to $t$ than $v_i$. 
So consider a topological ordering of the vertices, where edges are conceptually oriented from smaller to larger $d(v_i,t)$ values.
Thus if we compute the $\csp(v_i,t)$ values in increasing order of the index $i$, then each $\csp(v_i,t)$ value can be computed in time proportional to the degree of $v_i$, 
and so the overall running time is $O(m+n)$.
\end{proof}

\section{Improved Analysis for Complete Graphs}\label{apnd:completeGraphs}

Here we consider the special case when $G=K_n$, improving parts of the analysis from \cite{frv-mvdha-18,Gilbert2017}.  
First, we consider the $O(\opt^{1/3})$-approximation algorithm of \cite{frv-mvdha-18}, which works for both MR$(K_n,\mathbb{R})$ and MR$(K_n,\mathbb{R}_{\geq 0})$.
The running time of this algorithm is $\Theta(n^6)$, since at some point it enumerates all cycles of length $\leq 6$.
With a more careful analysis, we observe it suffices to consider cycles of length $\leq 5$, improving the running time to $\Theta(n^5)$.  For MR$(K_n,\mathbb{R}_{\geq 0})$ we consider a simple, appealing algorithm with good empirical performance from \cite{Gilbert2017}, referred to as {\sc IOMR-fixed}.  We prove that unfortunately it is an $\Omega(n)$ approximation. 

\subsection{5 Cycle Cover}\label{sec:5cyclecover}
Here we argue the running time of the $O(\opt^{1/3})$-approximation algorithm of \cite{frv-mvdha-18}, which works for both MR$(K_n,\mathbb{R})$ and MR$(K_n,\mathbb{R}_{\geq 0})$, 
can be improved from $\Theta(n^6)$ to $\Theta(n^5)$.  
The algorithm presented in \cite{frv-mvdha-18} has 3 major steps. The first two steps are used to approximate the support of the optimal solution and then the last step is actually used to find a solution given this support. We shall focus on the first 2 steps as these are where we make modifications. 

\begin{figure}[ht]
\centering
\includegraphics[width = 0.3\textwidth, height = 0.4\textwidth]{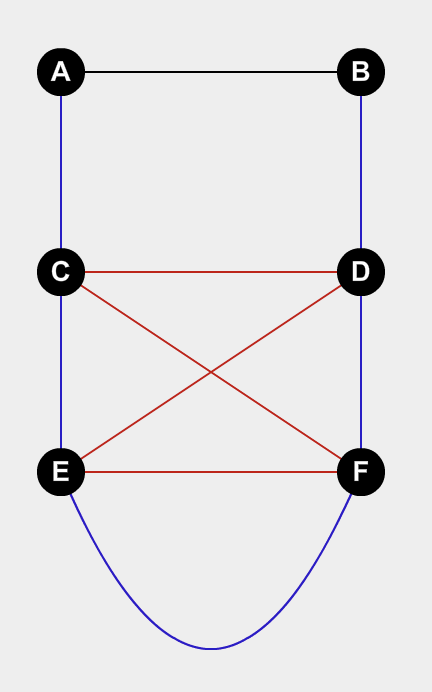} \hspace{2cm}
\includegraphics[width = 0.337\textwidth]{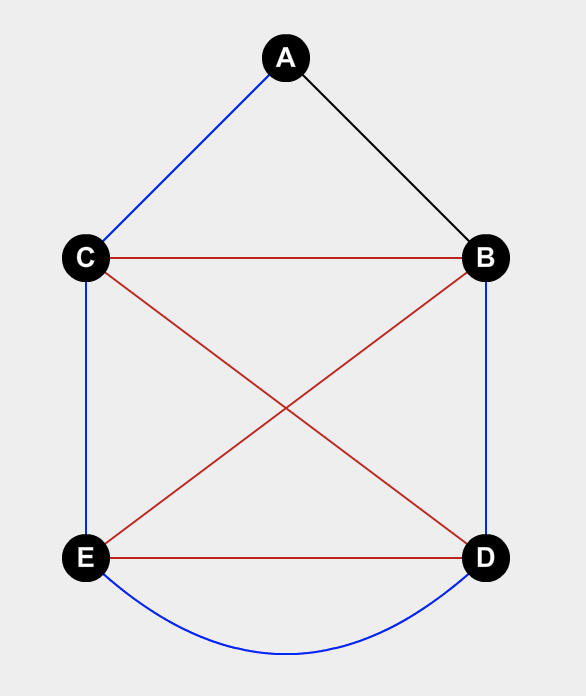}
\caption{Left: Embedding from \cite{frv-mvdha-18}. Right: Our modified embedding for a smaller cycle. Here the black edge is the heavy edge. The blue edges are the light edges and the red edges are the embedded 4 cycle. The curved blue edge indicates that there are more vertices along that path}
\label{fig:4cycle}
\end{figure}

\textbf{First Step:} In the first step,~\cite{frv-mvdha-18} find a cover for all broken cycles of length $ \le m$. In particular, the authors use the case when $m = 6$. As described in \cite{frv-mvdha-18}, we can obtain an $m-1$ approximation of the optimal cover for all broken cycles of length $ \le m$ in $O(n^m)$ time. Denote this cover by $S_{\le m}$. 

\textbf{Second Step:} For this step, we need to first define unit cycles. Given a broken cycle $C$ with heavy edge $h$, let $e$ be a chord of $C$. Then $e$ divides $C$ into 2 cycles, one that contains $h$, denoted heavy($C,e$) and one that does not contain $h$ denoted light($C,e$). We say this cycle is a unit cycle if for all chords $e$, $e$ is not the heavy edge of light($C,e$). 

From the definition of a unit cycle, a light cover of all unit cycles light covers all broken cycles. Hence, step 2 of the algorithm from \cite{frv-mvdha-18} light covers all unit cycles not covered by $S_{\le 6}$ as follows. Let $C$ be such a unit cycle. Now we know that $C$ has at least $7$ edges. Consider the red $C_4$ shown in Figure \ref{fig:4cycle}. We know that for each $e \in C_4$, we have that heavy$(C,e)$ is a broken cycle with at most 6 edges. Hence, we must have at least 1 edge in $S_{\le 6}$. But since $C$ has no light edges in $S_{\le 6}$, we must have $e \in S_{\le 6}$. Thus, we know all edges in $C_4$ are edges in $S_{\le 6}$. Moreover, observe that either chord of $C_4$ is a light edge of $C$.  Thus it suffices to compute a cover with least one chord of every four cycle from the edges in $S_{\le 6}$, a step which the authors in~\cite{frv-mvdha-18} denote $chord4(S_{\le 6})$. 

In Figure \ref{fig:4cycle}, we observe that the same 4 cycle can be embedded in a 6 cycle instead of a 7 cycle.
Thus, our modified algorithm is shown in Algorithm \ref{alg:fivecycle}.

\begin{algorithm}[ht]
\caption{5-Cycle Cover}
\label{alg:fivecycle}
\begin{algorithmic}[1]
\Function{5 Cycle Cover}{$G=(V,E,w)$}
	\State Compute a regular cover of $S_{\le 5}$ of all broken cycles with $\le 5$ edges
	\State Compute a cover $S_c = chord4(S_{\le 5})$
	\State \Return \Call{Verifier}{$G,S_c \cup S_{\le 5}$}
	\EndFunction
\end{algorithmic}
\end{algorithm}

\subsection{IOMR-fixed}\label{sec:iomr-fixed}

We will now show that {\sc IOMR-fixed} is an $\Omega(n)$ approximation algorithm. The algorithm presented in Gilbert and Jain \cite{Gilbert2017} is as follows:

\begin{algorithm}[ht]
\caption{IOMR Fixed}
\begin{algorithmic}[1]
\Require{$D \in \Sym$ }
\Function{IOMR-Fixed}{D}
\State $\hat{D} = D$
\For{$k \gets 1 \textrm{ to } n$}
\For{$i \gets 1 \textrm{ to } n$}
\State $\hat{D}_{ik} = \max(\hat{D}_{ik}, \max_{j < i}(\hat{D}_{ij}-\hat{D}_{jk}))$
\EndFor
\EndFor
\State \Return{$\hat{D}-D$}
\EndFunction
\end{algorithmic}
\end{algorithm}

. 

\begin{lemma} \label{lem:OmegaIOMR} For every $n$, there exists a weighted graph $G$ such that {\sc IOMR-Fixed} repairs $\binom{n-1}{2}$ edge weights while an optimal solutions repairs at most $(n-2)$ edge weights. 
\end{lemma}
\begin{proof} Consider a matrix $D$ where \[ D_{ij} = \begin{cases} 0 & \text{ if } i \neq 1, j \neq 1 \\ 2^i & \text{ if } j=1, i > 1 \\2^j & \text{ if } i=1, j > 1 \end{cases} \] This matrix $D$ will be the weight matrix for the input graph $K_n$. 

First, we claim that all entries of the form $D_{s1}$ will never be updated as entries will only be updated the first time they are seen. Thus 
\begin{align*} 
D_{s1} = \max(D_{s1}, \max_{t < s}(D_{s1} - D_{1t}))  
= \max(2^s, \max_{t < s} (2^s - 2^t)) 
 = 2^s 
\end{align*}

Now we just have to verify that the rest of the non-diagonal entries are updated. Let us look at the first time an entry $D_{rs}$ is updated. (Here $r < s$.) Then we have that 
\begin{align*} \hat{D}_{rs} &= \max(D_{rs}, \max_{t < s}(D_{st} - D_{tr})) 
=\max_{t < s}(D_{st} - D_{tr}) &[\text{Since }D_{rs}=0] \\
&\ge D_{s1}-D_{1r} 
= 2^s - s^r
> D_{rs}.
\end{align*} 
Thus all other non-diagonal entries will be updated the first time seen. Thus, for the solution $W = \hat{D} - D$ that {\sc IOMR-fixed} returns, we see that $W_{ij} > 0$ for exactly all $1 < i,j \le n$ and $i \neq j$. Thus, we repaired $\binom{n-1}{2}$ edge weights.  

Finally, a sparser increase only solution $W$ can be obtained as follows. For all $s>1$ we set \[ W_{1s} = W_{s1} = 2^n - D_{s1} \]  and all other entries of $W$ are $0$. This then gives us the desired result. \end{proof}

\begin{corollary} 
{\sc IOMR-fixed} is an $\Omega(n)$ approximation algorithm. 
\end{corollary}

\end{document}